%% file: main.tex
\documentclass[submission,copyright,creativecommons]{eptcs}
\providecommand{\event}{LSFA 2026} 

\usepackage{iftex}

\ifpdf
  \usepackage{underscore}         
  \usepackage[T1]{fontenc}        
\else
  \usepackage{breakurl}           
\fi

\usepackage{amsmath}
\usepackage{amsfonts}
\usepackage{amsthm}
\usepackage{mathtools}
\usepackage{stmaryrd}
\usepackage{rsfso}
\usepackage{xcolor}
\usepackage{xspace}
\usepackage[textwidth=0.65in,textsize=scriptsize]{todonotes}
\usepackage{mdframed}
\usepackage{multicol}
\usepackage{hyperref}
\hypersetup{
  linkcolor  = blue,
  citecolor  = blue,
  urlcolor   = blue,
  colorlinks = true,
}
\usepackage{comment}
\usepackage{tikz, tikzscale}
\usetikzlibrary{arrows,arrows.meta,calc,automata,positioning,decorations.pathreplacing,shapes.geometric,shapes.misc,graphs,backgrounds,shadows.blur}

\tikzset{align at top/.style={baseline=(current bounding box.north)}}
\tikzstyle{every node}=[font=\scriptsize]

\tikzset{
  every picture/.style = {
    thick,
    >=stealth',
    node distance = 2.5em and 3em,
  }
  ,
  cross line/.style = {
    preaction = {
      draw=white,
      -,
      line width=4pt
    }
  }
  , 
  state/.style = {
    circle,
    font = \footnotesize,
    draw,
    inner sep = 0pt,
    minimum size = 0.3em,
    fill
  }
  , 
  dot/.style = {
    fill,
    circle,
    inner sep=0mm,
    minimum size=1.25mm,
    line width=0mm
  }
  , 
  label-state/.style = {
    sloped,
    font = \scriptsize,
    label distance = -2pt
  }
  , 
  label-edge/.style = {
    font = \scriptsize,
    label distance = -2pt
  }
}

\input{macros.tex}

\title{Basic Model Theory for Path Predicate Modal Logic}
\author{Raul Fervari
\institute{CONICET and Universidad Nacional de C\'ordoba \\ Argentina}
\email{rfervari@unc.edu.ar}
\and
Santiago Figueira 
\institute{CONICET and Universidad de Buenos Aires \\ Argentina}
\email{santiago@dc.uba.ar}
\and
Gabriel Goren-Roig
\institute{CONICET and Universidad de Buenos Aires \\ Argentina}
\email{ggoren@dc.uba.ar}
\and
Leonardo Torres
\institute{IMDEA Software Institute \\
Spain}
\email{leonardo.torres@imdea.org}
}

\def\titlerunning{Basic Model Theory for Path Predicate Modal Logic}
\def\authorrunning{Fervari, Figueira, Goren-Roig, Torres}
\begin{document}
\maketitle

\begin{abstract}
Path Predicate Modal Logic ($\ppml$) is a generalization of Basic Modal Logic, where atoms are relational predicates instead of propositional symbols. The study of $\ppml$ is motivated as a way to abstractly investigate data-aware formalisms, such as XPath or DataGL. In this paper, we investigate some basic model theoretical aspects of $\ppml$ to better characterize its expressive power. More concretely, we investigate different ways of defining Hennessy--Milner classes, and a van Benthem characterization theorem. In doing so, we discuss the main challenges of dealing with the novel features of $\ppml$, and what are the similarities with the standard approaches. 
\end{abstract}

\input{intro}
\input{content}

\input{ultrafilter}
\input{characterization}

\input{final}

\bibliographystyle{eptcs}
\bibliography{references}
\end{document}

%% file: macros.tex
\newcommand{\defstyle}[1]{\textbf{#1}}
\newcommand{\arity}{\mathrm{arity}}
\newcommand{\last}{\mathrm{last}}
\newcommand{\ST}{\mathrm{ST}}
\newcommand{\card}[1]{|{#1}|}

\newcommand{\modelA}{\mathcal{A}}

\newcommand{\modelB}{\mathcal{B}}
\newcommand{\foModelM}{\mathfrak{M}}
\newcommand{\foModelN}{\mathfrak{N}}

\newcommand{\powerset}{\mathcal{P}}

\newcommand{\ue}{\mathfrak{ue}}
\newcommand{\ueA}{\modelA^\ue}
\newcommand{\ueB}{\modelB^\ue}
\newcommand{\bisimilar}{\mathrel{\underline{~~~}\hspace{-0.85em}{\leftrightarrow}}
}

\newcommand{\val}[2]{ \llbracket #1 \rrbracket^#2}
\newcommand{\ext}[1]{ \llbracket #1 \rrbracket}
\newcommand{\set}[1]{\{ #1 \}}

\newcommand{\lra}{\leftrightarrow}

\newcommand{\N}{\mathbb{N}}

\newcommand{\DEmp}{\mathrm{debt}}
\newcommand{\D}[1]{\DEmp(#1)}

\newtheorem{theorem}{Theorem}
\newtheorem{lemma}[theorem]{Lemma}

\newtheorem{corollary}[theorem]{Corollary}

\newtheorem{definition}[theorem]{Definition}
\newtheorem{proposition}[theorem]{Proposition}
\newtheorem{example}[theorem]{Example}

\renewcommand{\iff}{\text{\it iff}}
\newcommand{\IFF}{\quad\iff\quad}
\newcommand{\AND}{\text{ and }}

\newcommand{\moc}{\mathrm{PCon}}

\newcommand{\ppml}{\mathsf{PPML}}
\newcommand{\bml}{\mathsf{BML}}
\newcommand{\fol}{\mathsf{FOL}}
\newcommand{\alanguage}{\mathcal{L}}

\newcommand{\bigraul}[1]{\todo[inline,color=red!40]{{\bf R:} {\footnotesize #1}}}

\newcommand{\biggabo}[1]{\todo[inline,color=green!20]{{\bf G:} {\footnotesize #1}}}
\newcommand{\gc}{\biggabo}

\newcommand{\santi}[1]{\todo[color=cyan!20]{{\bf S:} {\footnotesize #1}}\xspace}
\newcommand{\bigsanti}[1]{\todo[inline,color=cyan!20]{{\bf S:} {\footnotesize #1}}}

%% file: intro.tex
\section{Introduction} 
\label{sec:intro}

Data-aware modal logics constitute a prominent approach to the study of languages arising in database theory. Informally, such logics are interpreted over relational structures equipped with data values at each point, and combine modalities for navigating the underlying graph with mechanisms for comparing data values. A central goal in this area is to obtain languages with a good balance between expressive power and computational behavior.

Among the most extensively studied examples are formal (mathematized) versions of the XML path language XPath (see, e.g.,~\cite{LibkinMV16}). This language is well suited for querying structured data, while, at the same time, being closely related to standard modal languages such as Propositional Dynamic Logic (PDL)~\cite{HKT00}. Its navigational fragment, Core-XPath~\cite{GKP05}, has been investigated in depth, both from the point of view of expressivity and from that of satisfiability and complexity (see, e.g.,~\cite{BK08,CateM09,FigueiraS11,Figueira12ACM,LibkinMV16,FigueiraD:2018}). Model-theoretic aspects of XPath-like languages have also been studied in, for example,~\cite{ICDT14Jair,ADF17IC,AbriolaBFF18}, while axiomatizations and completeness results can be found in~\cite{BaeldeLS19,AbriolaDFF17,AbriolaFG24,ACF25}.

These examples suggest the usefulness of more abstract frameworks, capable of isolating the logical principles behind data-aware navigation without being tied to a specific formalism or data format. Path Predicate Modal Logic ($\ppml$, for short), introduced in~\cite{FigGoren25}, is one such language. It is a simple extension of Basic Modal Logic ($\bml$)~\cite{Blackburn&deRijke&Venema01}, in which atomic propositions are not unary but rather $n$-ary predicates. Since such predicates are meant to hold of tuples of states occurring \emph{along a path}, formulas of $\ppml$ are interpreted over sequences of states rather than a single one; this is precisely the reason for the name \emph{Path Predicate Modal Logic}. In other words, one must keep track of the states visited along a path in order to evaluate an $n$-ary atom. As observed in~\cite{FigGoren25}, this provides a natural abstract setting for certain forms of data-aware expressivity; in particular, $\ppml$ can express the data-aware logic DataGL from~\cite{BaeldeLS19}. At the same time, $\ppml$ is of independent interest as a modal language for reasoning over structures based on richer signatures than those usually considered in $\bml$.

The paper~\cite{FigGoren25} introduced $\ppml$ and developed some aspects of its model theory as well as its relationships with other logics, including the data-aware logic DataGL, First Order Logic ($\fol$) and $\bml$ itself. In particular, it introduced and characterized a notion of bounded $\ppml$-bisimulation, together with a ``Hennessy--Milner'' result, and a standard translation into $\fol$. The general approach of said paper has two distinctive characteristics: firstly, it develops the model-theoretical aspects of $\ppml$ drawing heavily on the categorical framework of game comonads and arboreal categories. Thus it favors results that can be obtained directly from the categorical formalism, e.g. a tree-like model property, a Lovász-type homomorphism counting theorem and a Feferman-Vaught-Mostowski-style theorem for products of models; or results that at least relate somehow to this formalism, such as a version of the Chandra-Merlin correspondence for the positive fragment of $\ppml$. Secondly, it focuses on the so-called \emph{single-point semantics}, where formulas are interpreted at a single point of a structure, even though this required the definition of a more general \emph{valuation semantics} in which formulas are interpreted on a sequence of points, akin to $\fol$.

The current paper continues the development of the model theory of $\ppml$ in a way that is complementary to the previous development with respect to both characteristics. Firstly, it focuses on classical model-theoretical aspects of the logic that were neglected in the previous treatment: saturated models, ultrafilter extensions and a van Benthem-style characterization theorem~\cite{vanBenthem1983}. Secondly, it aims to study $\ppml$ in its more general valuation semantics and restates previous results in this, more general context.

\paragraph{Outline and Contributions.} In~Sec.~\ref{sec:ppml}, we introduce the syntax and semantics of $\ppml$, together with the valuation-based notions of standard translation of $\ppml$ into $\fol$ and unbounded bisimulation. Both definitions are adapted versions of those given in~\cite{FigGoren25}, better suited for our current purposes.  For these adapted versions, we start by proving some basic properties.  Precisely, we show correctness of the translation (Prop.~\ref{prop:equiv1orderPPML}) and bisimulation invariance (Thm.~\ref{th:bisImplyEquiv}), by straightforward modifications of the corresponding results from the literature~\cite[Prop. 4.1 and Thm. 2.10(2)]{FigGoren25}.

In~Sec.~\ref{sec:ufe}, we study the Hennessy--Milner property for $\ppml$. We introduce a suitable notion of $\ppml$-saturated structures (Def.~\ref{def:saturation}) and prove that both finitely-branching models and $\ppml$-saturated models constitute Hennessy--Milner classes for $\ppml$ (Thm.~\ref{th:HM} and Thm.~\ref{th:saturationImplHM}). On the other hand, we define ultrafilter extensions for models of $\ppml$ (Def.~\ref{def:uf-extension}), and prove a ``bisimilarity somewhere else'' theorem (Thm.~\ref{th:equivCharUF}) showing that logical equivalence between two models always coincides with bisimilarity between their ultrafilter extensions. These results follow closely the classical presentation for $\bml$ as in~\cite{Blackburn&deRijke&Venema01}, although they require a non-trivial adaptation of the classical ideas in order to deal with the intrinsic features of $\ppml$.

Finally, in ~Sec.~\ref{sec:characterization} we prove a van Benthem-style characterization theorem~\cite{vanBenthem1983} for $\ppml$, identifying this language with the $\ppml$-bisimulation-invariant fragment of $\fol$. We conclude with some remarks on this and future work in~Sec.~\ref{sec:final}.

%% file: content.tex
\section{$\ppml$ and Basic Constructions}
\label{sec:ppml}

In this section, we review the basic definitions of Path Predicate Modal Logic ($\ppml$), namely its syntax, models, and semantics, together with some auxiliary notions. In contrast with~\cite{FigGoren25}, we focus on the so-called valuation semantics for $\ppml$, i.e., the semantics in which formulas are interpreted with respect to a sequence of points, rather than under the single-point semantics. In doing so, we adapt the previously introduced notions of standard translation into First-Order Logic ($\fol$) and $\ppml$-bisimulation into a suitable form for valuation semantics, and observe that both the correctness of the standard translation and bisimulation invariance can be established using essentially the same arguments.

\subsection{Path Predicate Modal Logic}

\paragraph{Preliminaries.}

We assume standard notions from $\fol$ (see, e.g.,~\cite{Enderton1972}). We now fix some definitions and notation. A \defstyle{relational first-order signature}, or simply a \defstyle{signature}, consists of a set $\sigma$ of relation symbols together with a function $\arity : \sigma \to \mathbb{N}_{>0}$ assigning to each symbol a positive integer, its \defstyle{arity}. As usual, we often write simply $\sigma$ for the pair $(\sigma,\arity)$. A \defstyle{$\sigma$-structure}, or \defstyle{model}, is a tuple $\modelA=(A,\cdot^{\modelA})$, where $A$ is a nonempty set, called the \defstyle{domain} of $\modelA$, and for each relation symbol $R\in \sigma$, the interpretation $R^{\modelA}$ is a subset of $A^{\arity(R)}$. We will shortly introduce a designated binary relation symbol $E\in \sigma$; given a model $\modelA$ we will write $a_1\prec a_2$ whenever $(a_1,a_2)\in E^{\modelA}$.

Given a set $S$, we write $S^+$ for the set of all nonempty finite sequences over $S$. We write sequences as the juxtaposition of their elements without any additional symbol, and similarly for the concatenation of sequences. We identify sequences with tuples when necessary. If $s=s_1\ldots s_{\card{s}}\in S^+$, we denote by $\last_k(s)$ the suffix of length $k$ of $s$ whenever $k\leq \card{s}$, and we set $\last_k(s)=s$ otherwise. We write simply $\last(s)$ for the last element of $s$, namely $s_{\card{s}}$. 

Finally, we call a \defstyle{valuation} a sequence $s\in A^+$, where $A$ is the domain of a model $\modelA$.

\paragraph{Syntax and Semantics of $\ppml$.}

Let $\sigma$ be a relational signature containing a distinguished binary relation symbol $E$. The set of \defstyle{$\ppml$-formulas} over $\sigma$ is generated by the grammar
\[
    \phi,\psi ::= \top \mid R \mid \lnot \phi \mid \phi \land \psi \mid \Diamond \phi,
\]
where $R\in \bar{\sigma}$ and $\bar{\sigma}:=\sigma\setminus\{E\}$. Thus, the distinguished relation $E$ is not itself an atomic formula of the language. The remaining Boolean connectives ($\bot,\vee,\rightarrow,\leftrightarrow$) and the modality $\Box$ are defined as usual.

Given a $\sigma$-structure $\modelA$ and a valuation $s\in A^+$, the satisfaction relation $\modelA,s\models \phi$ is defined inductively as follows:
\[
\begin{array}{l@{\quad}c@{\quad}l}
    \modelA,s \models \top
    & &
    \text{always}, \\[1mm]
    \modelA,s \models R
    &\iff&
    \arity(R)\leq \card{s} \text{ and } \last_{\arity(R)}(s)\in R^{\modelA}, \\[1mm]
    \modelA,s \models \lnot \phi
    &\iff&
    \modelA,s \not\models \phi, \\[1mm]
    \modelA,s \models \phi \land \psi
    &\iff&
    \modelA,s \models \phi \text{ and } \modelA,s \models \psi, \\[1mm]
    \modelA,s \models \Diamond \phi
    &\iff&
    \text{there exists } a\in A \text{ such that } \last(s)\prec a \text{ and } \modelA,sa \models \phi.
\end{array}
\]
The crucial difference with respect to the semantics of Basic Modal Logic ($\bml$) is that atomic formulas are now interpreted by means of relations of arbitrary arity, and the modality $\Diamond$ extends the current sequence by appending a successor. In this way, the valuation records the history of the path constructed during evaluation.

Given a formula $\phi$ and a model $\modelA$, we write
\[
    \val{\phi}{\modelA}:=\{s\in A^+ \mid \modelA,s\models \phi\}
\]
for its \defstyle{extension} in $\modelA$. We say that $\phi$ is \defstyle{satisfiable} if there exists some $\sigma$-structure $\modelA$ such that $\val{\phi}{\modelA}\neq\emptyset$. If $\modelA$ and $\modelB$ are $\sigma$-structures and $s\in A^+$, $t\in B^+$ are sequences such that $\card{s}=\card{t}$, we say that $s$ and $t$ are \defstyle{$\ppml$-equivalent} if they satisfy the same $\ppml$-formulas, that is, if for every $\psi\in \ppml$,
\[
    \modelA,s\models \psi \quad \iff \quad \modelB,t\models \psi.
\]
In that case we write $(\modelA,s)\equiv (\modelB,t)$, or simply $s\equiv t$ when the models are clear from the context.

\begin{example}\label{example}
    Consider the signature $\sigma = \set{E,B}$, where $\arity(E)= \arity(B)=2$. Let $\modelA = (A, E^\modelA, B^\modelA)$ be the following $\sigma$-structure: its domain is $A = \set{a,b,c,d}$, the accessibility relation $E^\modelA$ is represented by the black edges in the following graph, while $B^\modelA$ is given by dashed red edges.

    \begin{center}
    \begin{tikzpicture}
        \node[state] (b) [label=below:{$b$}] {};
        \node[state] (a) [label=above:{$a$}] [above left = of b] {};
        \node[state] (c) [label=below:{$c$}] [right = of b]{};
        \node[state] (d) [label=below:{$d$}] [below = of a] {};

        \path (a) edge[->] node [left] {} (b);
        \path (b) edge[->] node [left] {} (c);
        \path (b) edge[->] node [left] {} (d);
        \path (a) edge[->] node [left] {} (d);

        \path[->,dashed, red!80, thin, bend left=60]  (a)  edge (b);
        \path[->,dashed, red!80, thin, bend left=60]  (b)  edge (d);
    \end{tikzpicture}
    \end{center}
    
    \noindent Observe that $\modelA, a \models \Diamond (B \land \Diamond B)$, since $abd$ form an $E^\modelA$-path, and moreover $(a,b), (b,d) \in B^\modelA$.
\end{example}

Notice that $\ppml$ extends $\bml$. Indeed, if every relation symbol in $\bar{\sigma}$ has arity $1$, atomic formulas depend only on the last element of the current sequence, so the extra path information becomes irrelevant and the semantics coincides with the standard Kripke semantics of $\bml$. Thus, over so-called \textbf{unimodal signatures}, $\ppml$ is precisely Basic Modal Logic.

The semantic interpretation above (i.e., with respect to \emph{sequences} of points) will  henceforth be called \defstyle{valuation semantics}, whereas the \defstyle{single-point semantics} consists of the restriction of the satisfaction relation $\models$ to valuations of length $1$. In other words, the single-point semantics of a formula $\phi$ is given by its interpretation in \defstyle{pointed models} $(\modelA, a)$ with $a \in A$. As mentioned before, contrary to~\cite{FigGoren25} we will consider $\ppml$ with its valuation semantics.

The modality $\Diamond$ plays the role of extending the current valuation, and therefore builds the path on which relation symbols are eventually evaluated. In particular, an occurrence of a relation symbol $R$ of arity $n$ can only be meaningfully evaluated once the current valuation has length at least $n$; otherwise, it is automatically false. For instance, let $\phi=\Diamond\Diamond\Diamond R \lor R$ with $\arity(R)=3$. Starting from a valuation of length $2$, the left disjunct can be evaluated, since the three diamonds extend the valuation into a sequence of length $5$, sufficient to interpret $R$. By contrast, the right disjunct is false independently of the model, because a valuation of length $2$ is too short to evaluate $R$.

This suggests measuring, for each formula $\phi$, the minimal length of a ``sensible'' valuation for $\phi$, i.e. a valuation such that each atom in the formula is evaluated on a valuation of sufficient length.

\begin{definition}\label{def:debt}
Let $\sigma$ be a relational signature and let $\phi \in \ppml$. We define the function $\D{\phi}$ recursively as follows:
\[
\begin{array}{r@{\qquad}r}
\begin{array}{lcl}
    \D{\top} & := & 0, \\
    \D{R} & := & \arity(R) - 1, \\
    \D{\neg \phi} & := & \D{\phi},
\end{array}
&
\begin{array}{lcl}
    \D{\phi_1 \land \phi_2} & := & \max\{ \D{\phi_1}, \D{\phi_2} \}, \\
    \D{\Diamond \phi} & := & \max\{ 0, \D{\phi} - 1 \}.
\end{array}
\end{array}
\]
We refer to $\D{\phi}$ as the \textbf{modal debt} of $\phi$.
\end{definition}

In~\cite{FigGoren25}, a $\ppml$ formula $\phi$ was called \textbf{well nested} whenever $\D{\phi} = 0$. It makes sense to single-out the well-nested formulas in the context of single-point semantics, but not so much for valuation semantics. Indeed, the formula $\phi = R$ for some symbol $R$ of arity $n > 1$ may be satisfiable under valuation semantics, even though it is not satisfiable under single-point semantics. In this sense, our new context suggests reinterpreting modal debt not as a measure of how far a formula is from being ``sensible'' (i.e. how much additional path information is still needed for its evaluation), but rather as controlling the length of ``sensible valuations'' for the formula: a valuation $s$ is ``sensible'' for $\phi$ iff $|s| \geq \D{\phi} + 1$. Notice that given $\phi$ and a valuation $s$ in a model $\modelA$ with $|s| \leq \D{\phi}$, it is not necessarily the case that $\modelA, s \not\models \phi$. The simplest such example is $\phi = \lnot R$ for $R$ with $\arity(R) > 1$, then $\modelA, s \models \phi$ for all $s$ with $|s|< \arity(R)$.


We close the section with a novel normalization result which clarifies this non-trivial behavior of the semantics of $\ppml$. Moreover, it will constitute an essential technical tool in the upcoming sections. The idea is the following: if we want to evaluate a formula $\phi$ only on valuations of a fixed length $n$, then we are free to rewrite \emph{some} of the atoms in $\phi$ into $\bot$, namely the atoms of arity $r$ which are not in the scope of at least $r-n$ diamonds. This always results in a formula $\psi$ with $\D{\psi}\leq n-1$.

More precisely, for each $n\geq 1$, define recursively a map $N_n:\ppml\to\ppml$ by
\[
N_n(\top):=\top,
\qquad
N_n(R):=
\begin{cases}
R & \text{if } \arity(R)\leq n,\\
\bot & \text{otherwise},
\end{cases}
\]
\[
N_n(\neg\phi):=\neg N_n(\phi),
\qquad
N_n(\phi\land\psi):=N_n(\phi)\land N_n(\psi),
\qquad
N_n(\Diamond\phi):=\Diamond N_{n+1}(\phi).
\]
%
On the other hand, for any pair of $\sigma$-structures $\modelA, \modelB$ and valuations $s\in A^+, t\in B^+$, let us write $$\modelA, s \equiv_n \modelB, t$$ whenever $|s|=|t|$ and for every $\phi \in \ppml$ with $\D{\phi}\leq n-1$,
$\modelA, s \models \phi \ \iff \ \modelB, t \models \phi.$
We also write $s \equiv_n t$ when the structures are clear from context.

\begin{lemma}\label{lem:debt-normalization}
Let $\phi\in \ppml$ and $n\geq 1$. Then there exists $\psi\in\ppml$ with $\D{\psi}\leq n-1$ such that for any $\sigma$-structure $\modelA$ and valuation $s \in A^+$ with $|s| = n$,
$$\modelA, s \models \phi \quad \iff \quad \modelA, s \models \psi.$$

As a consequence of this fact, for any pair of $\sigma$-structures $\modelA, \modelB$ and valuations $s\in A^+, t\in B^+$ with $|s|=|t|=n$,
$$\modelA, s \equiv \modelB, t \quad \iff \quad \modelA, s \equiv_n \modelB, t.$$
\end{lemma}
\begin{proof}
    For the first claim, take $\psi = N_n(\phi)$. The proof is then straightforward by structural induction on $\phi$. Now let $\modelA, \modelB$ be a pair of $\sigma$-structures and let $s\in A^+, t\in B^+$ with $|s|=|t|=n$. Clearly $s \equiv t$ implies $s \equiv_n t$. As for the reverse implication, assuming $s \equiv_n t$, let $\phi \in \ppml$ be arbitrary. Then taking $\psi$ as above
    $$
    s \models \phi \quad\iff\quad s \models \psi \quad\iff\quad t\models \psi \quad\iff\quad t \models \phi,
    $$
    and this concludes the proof.
\end{proof}



\subsection{The Standard Translation} \label{sec:st}

In what follows we fix a set of first-order variables $x_1,x_2,\dots$ indexed by the natural numbers.

\begin{definition}[Standard Translation]\label{def:st}
Given a sequence of first-order variables $\overline{x} = x_1\ldots x_n$ for some $n \in \mathbb{N}_{>0}$ and a $\ppml$ formula $\phi$, we define a $\fol$ formula $\ST_{\overline{x}}(\phi)$ over the signature $\sigma$ with free variables in $\overline x$ by structural induction:
\[
\begin{array}{lcl}
    \ST_{\overline{x}}(\top) &:= &  \top \\
    \ST_{\overline{x}}(R) &:= & 
        \begin{cases}
            R(\last_{\arity(R)}(\overline{x})) & \text{if } \arity(R) \leq \card{\overline{x}} \\
            \bot & \text{otherwise}
        \end{cases} \qquad (R \in \bar{\sigma}) \\
    \ST_{\overline{x}}(\lnot\phi) & := & \lnot \ST_{\overline{x}}(\phi) \\
    \ST_{\overline{x}}(\phi \land \psi) &:= & \ST_{\overline{x}}(\phi) \land \ST_{\overline{x}}(\psi) \\
    \ST_{\overline{x}}(\Diamond \phi) & := & \exists x_{n+1} (E(x_n,x_{n+1}) \land \ST_{\overline{x}x_{n+1}}(\phi)).
\end{array}
\]
\end{definition}

This definition adapts the one from~\cite{FigGoren25} for the valuation semantics. Indeed, therein, by focusing on single-point semantics, it is shown that the standard translation actually falls in the fragment of $\fol$ using at most $N$ variables, for some $N \in \N$,  whenever the arity of symbols in $\sigma$ is bounded by $N$.\footnote{It is also easy to see that the $\ppml$ fragment of formulas with \defstyle{modal depth} at most $k$, i.e. with at most $k$ nested diamonds, falls inside the $\fol$ fragment with quantifier rank bounded by $k$. This is true both in the context of single-point and of valuation semantics.} This is not the case when considering full valuation semantics, hence their definition of a standard translation with a cyclic set of $N$ variables is not applicable for our purposes.

Having introduced this variant of the definition of $\ST$, we readily obtain the following result, which in a sense generalizes the one of \cite{FigGoren25} to valuation semantics.

\begin{proposition}\label{prop:equiv1orderPPML}
Let $\modelA$ be a $\sigma$-structure. 
For every formula $\phi \in \ppml$ and sequence of variables $\bar x$ with $|\bar x| \geq \D{\phi} + 1$, we have 
\[   
\modelA,s \models \phi \quad \iff \quad \modelA \models \ST_{\bar x}(\phi)[s].
\]
\end{proposition}
\begin{proof}
    Straightforward by structural induction on $\phi$.
\end{proof}

We will come back to the standard translation in Sec. \ref{sec:characterization}, where we will use it to identify $\ppml$ with a bisimulation-invariant fragment of $\fol$.



It is natural to ask how the image of the standard translation $\ST$ relates to well-studied fragments of $\fol$. In the case of $\bml$, the classical answer is that the standard translation lands inside the {\em guarded fragment} (GF), introduced in \cite{AndrekaBenthemNemeti98}. Recall that GF is obtained by restricting existential quantification to formulas of the form
$
 \exists\bar y (\alpha(\bar x, \bar y) \land \varphi(\bar x,\bar y)),
$
where $\alpha$ is an atomic formula containing all free variables of the matrix. The usual standard translation of $\bml$ has exactly this shape:
$
\ST_x(\Diamond\phi) = \exists y (E(x,y)  \land \ST_y(\phi)).
$
Thus, the modal step is guarded by the accessibility atom $E(x,y)$, and the image of the
 standard translation of $\bml$ is contained in GF.

The same remains true for $\ppml$ as long as all non-modal relation symbols have arity at most 2. Indeed, in that case unary atoms are translated into formulas involving only the last variable of the current valuation, and binary atoms into formulas involving only the last two variables. Hence, by a straightforward induction on $\phi$, every formula $\ST_{x_1 \cdots x_n}(\phi)$ depends only on the suffix $x_{n-1},x_n$ once $n\geq 2$. It follows that the modal clause
$
\ST_{x_1\cdots x_n} (\Diamond \phi)  = \exists x_{n+1}(E(x_n,x_{n+1}) \land \ST_{x_1\cdots x_{n+1}} (\phi))
$
is
guarded.

A natural weakening of GF is the {\em loosely guarded fragment} (LGF), introduced by Gr\"adel in \cite{Gradel99}. In LGF, the guard may be a conjunction of atoms rather than a single atom, provided that every pair of free variables in the matrix appears together in some atomic conjunct of the guard. From the point of view of $\ppml$, however, this relaxation is still not enough once one allows relation symbols of arity 3.
To see this, let $R$ be a ternary relation symbol and consider the $\ppml$-formula
$\psi = \Diamond\Diamond \neg R$.
Its standard translation with one free variable is
$
\theta(x) =\exists y  (E(x,y) \land \exists z(E(y,z) \land \neg R(x,y,z))).
$
Semantically equivalent to $\exists y,z(E(x,y) \land E(y,z) \land \neg R(x,y,z))$, $\theta(x)$ says that there is an $E$-path of length 2 starting at $x$ such that the ternary relation $R$ fails on the triple of visited states.
%
%
It can be shown that $\theta(x)$ is not equivalent to any formula of LGF, showing that $\ppml$ goes beyond LGF once relation symbols of arity at least 3 are allowed.
The key point is already visible in the shape of the formula: the variables $x,y,z$ all occur in the matrix, but there is no guard covering the pair $(x,z)$.

A more appropriate  ambient fragment for $\ppml$ is the {\em fluted fragment}, originally due to Quine and studied extensively in modern form in~\cite{PrattHartmannSzwastTendera2019}. Roughly speaking, fluted formulas are those in which the order of quantification agrees with the order in which variables occur as arguments of predicates. This matches the present setting rather well. On the one hand, if $\phi$ is an atomic $\ppml$-formula of the form $R$, then 
$\ST_{x_1 \cdots x_n}(R) = R(\last_{\arity(R)}(x_1 \cdots x_n))$, 
so atoms are translated into predicates applied to contiguous suffixes of the current valuation. On the other hand, modal formulas extend the current context by one fresh last variable:
$
\ST_{x_1\cdots x_n} (\Diamond \phi)
   = \exists x_{n+1}  (E(x_n,x_{n+1}) \land \ST_{x_1 \cdots x_{n+1}}(\phi)).
$
Hence, by a straightforward induction on $\phi$, the image of the standard translation of $\ppml$ is contained in the fluted fragment of~$\fol$. 

In this sense, while $\bml$ naturally lives inside a guarded environment, $\ppml$ naturally lives in a fluted one, both in its single-point and valuation semantics. This inclusion is proper. Just as $\bml$ is not the whole guarded fragment, $\ppml$ is not the whole fluted fragment: the image of its standard translation is constrained by a specific path discipline, namely that atoms only inspect contiguous suffixes of the current valuation and that quantification proceeds only by appending a fresh last variable via the distinguished relation $E$. The full fluted fragment is strictly more permissive.

%

\subsection{Bisimulation}
Now let us turn our attention to bisimulations, which will be central to our results in both Sec.~\ref{sec:ufe} and Sec.~\ref{sec:characterization}. As usual, bisimulation provides a structural characterization of logical equivalence without referring to syntax explicitly.

\begin{definition}\label{def:bisimulations}
Let $\modelA$ and $\modelB$ be two $\sigma$-structures. Let $Z \subseteq A^+ \times B^+$ be a nonempty relation that relates sequences of the same length, this is,
$Z \subseteq \bigcup_{i \in \mathbb{N}} A^i \times B^i$.
We say that $Z$ is a \defstyle{$\ppml$-bisimulation} between $\modelA$ and $\modelB$ if the following conditions hold:
\begin{description}
    \item[(pred)] If $(s,t) \in Z$, then for every $R \in \bar{\sigma}$ we have
    $\modelA,s \models R$ iff $\modelB,t \models R$.
    
    \item[(zig)] If $(s,t) \in Z$, then for every $a \in A$ s.t. $\last(s) \prec a$, there exists $b \in B$ s.t. $\last(t) \prec b$ and $(sa,tb) \in Z$.

    \item[(zag)] If $(s,t) \in Z$, then for every $b \in B$ s.t. $\last(t) \prec b$, there exists $a \in A$ s.t. $\last(s) \prec a$ and $(sa,tb) \in Z$.
\end{description}
If $Z$ is a bisimulation relating the sequences $s$ and $t$ in the models $\modelA$ and $\modelB$, respectively, we say that $s$ and $t$ are bisimilar and write $(\modelA,s) \bisimilar (\modelB,t)$, or simply $s \bisimilar t$ when the context is clear.
\end{definition}

This notion of $\ppml$-bisimulation is adapted from~\cite{FigGoren25} in the following two ways: firstly, it relates sequences instead of points in accordance with our focus on valuation semantics; secondly, it is a notion of \emph{unbounded} bisimulation, in contrast with the previously studied bounded bisimulation. The unbounded notion is more appropriate for our current investigation since here we do not deal with $\ppml$ fragments of bounded modal depth.

Notice that the notion of $\ppml$ bisimulation presented above gives an unbounded bisimulation for single-point semantics by simply restricting the relation $\bisimilar$ to valuations of length $1$.
Moreover, if $\sigma$ is a unimodal signature then $s \bisimilar t$ is always equivalent to $\last(s) \bisimilar \last(t)$ and this notion of bisimulation reduces to the standard, unbounded $\bml$ bisimulation, just like the one in~\cite{FigGoren25} reduces to the standard bounded $\bml$ bisimulation.

We now prove bisimulation invariance of truth for $\ppml$ formulas with valuation semantics by an immediate adaptation of the corresponding proof in \cite{FigGoren25}.

\begin{theorem}\label{th:bisImplyEquiv}
    Let $\modelA$ and $\modelB$ be two $\sigma$-structures, and $s$ and $t$ sequences from the respective domains. Then $(\modelA,s) \bisimilar (\modelB,t)$ implies $(\modelA,s) \equiv (\modelB,t)$.
\end{theorem}
\begin{proof}
    Let $\modelA$ and $\modelB$ be two $\sigma$-structures such that $(\modelA,s) \bisimilar (\modelB,t)$. The proof proceeds by induction on the structure of $\varphi$.
    The case for the atomic formulas is straightforward by definition, while the cases of Boolean operators are standard. Then, we focus on the case for formulas of the form $\Diamond\psi$. Suppose $ \modelA,s \models \Diamond \psi$. By definition this means that there is $a \in A$ such that $\last(s) \prec a$ and $\modelA,sa \models \psi$. By  (zig), there is $b \in B$ such that $\last(t) \prec b$ and $(\modelA,sa) \bisimilar (\modelB,tb)$. Then, by inductive hypothesis, we have that $(\modelB,tb) \models \psi$, which allows us to conclude that $(\modelB,t) \models \Diamond \psi$. 
    The converse implication is analogous, using (zag) instead.
\end{proof}

%% file: ultrafilter.tex
\section{Hennessy--Milner Classes and Ultrafilter Extensions}
\label{sec:ufe}

\subsection{Hennessy--Milner classes}

As is well-known in the case of Basic Modal Logic, the converse of~Thm.~\ref{th:bisImplyEquiv} does not hold in general. A \defstyle{Hennessy--Milner class} is a class of structures for which it does. In our setting, this means that for any $\modelA$ and $\modelB$ in the class and for any $s \in A^+$ and $t \in B^+$, $(\modelA, s) \equiv (\modelB, t)$ implies $(\modelA, s) \bisimilar (\modelB ,t)$. A classical example of a Hennessy--Milner class for $\bml$ is that of \defstyle{finitely branching structures}, i.e., structures such that for every point $a$, the number of successors of $a$ is finite~\cite{Hennessy&Milner}. As shown in~\cite{FigGoren25}, this class also works for $\ppml$ with single-point semantics (where successors are taken with respect to the special relational symbol $E$). We now observe that this result generalizes to valuation semantics. 

\begin{theorem} \label{th:HM}
Let $\modelA$ and $\modelB$ be finitely branching $\sigma$-structures and let $s\in A^+$ and $t\in B^+$. Then $(\modelA,s) \equiv (\modelB,t)$ implies $(\modelA,s) \bisimilar (\modelB,t)$.
\end{theorem}
\begin{proof}
The idea of the proof is to show that the relation $\equiv$ is a bisimulation.  
The (pred) case is direct. 
For the (zig) case, let $s,t$ be sequences of $A^+$ and $B^+$ respectively, such that $s \equiv t$ and let $a \in A$ such that $\last(s) \prec a$. 

Aiming for a contradiction, suppose that there is no $b$ such that $\last(t) \prec b$ and $sa \equiv tb$. Let $\Delta$ be the set of successors of $\last(t)$. By hypothesis, this set is finite. Moreover, it is nonempty, since we have $\modelA, s \models \Diamond \top$ and, by inductive hypotheses, $\modelB, t \models \Diamond \top$. Let $\Delta=\{b_1,\ldots,b_n\}$. Since there is no successor $b$ such that $sa \equiv tb$, there must be a formula $\psi_i$ such that $\modelA,sa  \models \psi_i$ and $\modelB,tb_i \not \models \psi_i$ for each $i \in \{1,\ldots,n\}$. But then we have that $\modelA, s \models \Diamond(\bigwedge\limits_{i=1}^n \psi_i)$, while also $\modelB, t \not \models \Diamond(\bigwedge\limits_{i=1}^n \psi_i)$, contradicting the assumption $s \equiv t$.  
The (zag) case is analogous to (zig).
\end{proof}

Arguably, finite branching is quite a strong restriction on models, motivating the  developments in the forthcoming sections. We will consider two different approaches: on one hand, we will characterize a more general Hennessy--Milner class, whose members we call $\ppml$-saturated structures. On the other hand, we will present an ultrafilter extension construction which allows us to transform models into others having the desired property.

\subsection{Saturated structures}

We now introduce a notion of saturated $\sigma$-structures which is appropriate for $\ppml$. This is a straightforward variation of so-called \emph{modally-saturated} or \emph{m-saturated} models~\cite[Def. 2.53]{Blackburn&deRijke&Venema01}. In what follows, given a $\sigma$-structure $\modelA$, a binary relation $R$ over $A$ and a sequence $s\in A^+$, we write
\[
    R(s) \coloneqq \set{sa \in A^+ \mid (\last(s), a)\in R}.
\]
for the set of ``successor sequences'' with respect to the relation $R$.

\begin{definition}\label{def:saturation}
Let $\modelA$ be a $\sigma$-structure. We say that a set of formulas $\Sigma$ is \defstyle{satisfiable by a subset $X \subseteq A^+$} if there exists a sequence $s \in X$ such that $\modelA,s \models \phi$ for every $\phi \in \Sigma$. We say $\Sigma$ is \defstyle{finitely satisfiable by a subset $X$} if every finite subset of $\Sigma$ is satisfiable by $X$.

We say that $\modelA$ is \defstyle{$\ppml$-saturated} if for every sequence $s \in A^+$ and every set $\Sigma$ of $\ppml$ formulas, if $\Sigma$ is finitely satisfiable by $E^\modelA(s)$ then it is satisfiable by $E^\modelA(s)$.
\end{definition}

As observed in~\cite{Blackburn&deRijke&Venema01}, saturation indicates a sort of \emph{compactness property}: if a set of formulas can be finitely satisfied by the set of successors of a sequence, then the whole set can also be satisfied there. This compactness property implies the Hennessy--Milner property for the class, as we now show.

\begin{theorem} \label{th:saturationImplHM}
The class of $\ppml$-saturated $\sigma$-structures has the Hennessy--Milner property.
\end{theorem}

\begin{proof}
The proof proceeds similarly as for~Thm.~\ref{th:HM}, by showing that equivalence between $\ppml$-saturated structures is a bisimulation relation.

Let $\modelA$ and $\modelB$ be $\ppml$-saturated $\sigma$-structures, and let $s \in A^+$ and $t \in B^+$. 
The (pred) condition is trivial by the definition of equivalence. 
We will prove that $\equiv$ satisfies the (zig) condition of the definition of bisimulation. Suppose that $s \equiv t$ and that $\last(s) \prec a$. Let $\Sigma = \set{\phi \mid \modelA, sa \models \phi }$. 
Note that for every finite $\Delta \subseteq \Sigma$ we have $\modelA, s \models  \Diamond \left(\bigwedge \Delta\right)$. By applying inductive hypothesis, we have $\modelB, t \models \Diamond \left(\bigwedge \Delta\right)$, which implies that $\modelB, tb_{\Delta} \models  \bigwedge \Delta$ for some $b_{\Delta}$ such that $\last(t) \prec b_{\Delta}$. This means that $\Sigma$ is finitely satisfiable in $E^\modelB(t)$, then, by $\ppml$-saturation, there exists $tb \in B^+$ such that $tb \models \phi$ for every $\phi \in \Sigma$, hence $sa \equiv tb$. 
The proof of (zag) is  analogous.
\end{proof}

\subsection{Ultrafilter Extensions}

We now develop a notion of ultrafilter extension appropriate for $\ppml$. This model-theoretical construction essentially completes a model so as to make it saturated. Although we follow closely the classical presentation for $\bml$ as in e.g.~\cite{Blackburn&deRijke&Venema01}, we will need to deal with the intrinsic particularities of $\ppml$, which we will discuss along the way.


\begin{definition}\label{def:completion-ops}
Let $R$ be a binary relation on a set $A$. We define the following operations on $\powerset(A^+)$:
\[
\begin{array}{lcl}
    m_R(X) & \coloneqq & \{ s \in A^+ \mid \text{there is some } a \in A \text{ such that } sa \in X  \text{ and }(\last(s),a) \in R \} \\
    & = & \set{s\in A^+ \mid R(s) \cap X \neq \emptyset} \\
    l_R(X) & := & \{ s \in A^+ \mid \text{for all } a \in A, \text{ if } (\last(s),a) \in R, \text{ then } sa \in X \}\\
    & = & \set{s\in A^+ \mid R(s) \subseteq X}.
\end{array}
\]
\end{definition}

In the BML case, when interpreted over a model, $m_R(X)$ represents the set of points that \emph{can see} a state in $X$, while $l_R(X)$ are those that \emph{only see} points from $X$, in the model.
In the case of $\ppml$ the interpretation is similar, replacing points with sequences of points.



The next proposition states formally some intended facts: first, that $m_{E^{\modelA}}$ provides an alternative way to characterize the extension of a formula; second, that $m_R$ and $l_R$ are dual operators; and third, that $l_R$ distributes over intersections, analogously to the typical behaviour between $\Box$-like modalities and conjunctions.

\begin{proposition}\label{prop:l&mProperties}
    Let $\modelA=(A,\cdot^{\modelA})$ be a $\sigma$-structure, and let $R$ be a binary relation over $A$. Then for every $X\subseteq A^+$ and every collection $(X_i)_{i\in I}$ of subsets $X_i \subseteq A^+$, the following properties hold:
\begin{enumerate}
    \item $\val{\Diamond \phi}{\modelA} = m_{E^\modelA}( \val{\phi}{\modelA}) $,
    \item  $l_R(X) = A^+ \setminus m_R(A^+ \setminus X)$, and
    \item $\bigcap_{i\in I} l_R(X_i) = l_R (\bigcap_{i \in I} X_i)$.
\end{enumerate}
\end{proposition}

\begin{proof}\phantom{}

\begin{enumerate}
\item
$s \in \val{\Diamond \phi}{\modelA}$ iff $s \models \Diamond \phi$, iff there exists $s' \in E^\modelA(s)$ such that $s' \in \val{\phi}{\modelA}$, iff $s \in m_{E^\modelA}(\val{\phi}{\modelA})$.
\item $s\in l_R(X)$ iff $R(s) \subseteq X$ iff $R(s) \cap (A^+ \setminus X) = \emptyset$ iff $s \not\in m_R(A^+ \setminus X)$.
\item $s\in \bigcap_{i\in I} l_R(X_i)$ iff $R(s) \subseteq \bigcap_i X_i$ iff $s \in l_R (\bigcap_{i\in I} X_i)$.\qedhere
    \end{enumerate}
\end{proof}


Now we introduce the main construction of this section: ultrafilter extensions. In order to do so, let us briefly recall the notions of filter and ultrafilter (see, e.g.,~\cite{willard2012} for details). For a set $S\neq\emptyset$, a \defstyle{filter} over $S$ is a set $F\subseteq\powerset(S)$ such that the following closure properties hold: (i) $S\in F$, (ii) $X,Y\in F$ implies $X\cap Y\in F$, and (iii) $X\in F$ and $X\subseteq Y\subseteq S$ imply $Y\in F$. Whenever $F\neq\powerset(S)$, we say $F$ is \defstyle{proper}. If in addition for all $X\in\powerset(S)$, it holds that $X\in F$ iff $S\setminus X \not\in F$, then $F$ is an \defstyle{ultrafilter}.

\begin{definition}[Ultrafilter extension]\label{def:uf-extension}
Let $\modelA$ be a $\sigma$-structure. The \defstyle{ultrafilter extension} of $\modelA=(A,\cdot^{\modelA})$, denoted by $\ueA$, is the $\sigma$-structure defined as follows:
\begin{enumerate}
    \item Its domain is the set of ultrafilters over $A^+$.
    \item $(u,u^\prime) \in E^{\ueA}$ if and only if for all $X \in u^\prime$ we have $m_{E^{\modelA}}(X) \in u$. 
    \item For $R \in \bar{\sigma}$ of arity $n$, $(u_1, \ldots, u_n) \in R^{\ueA}$ if and only if $\val{R}{\modelA} \in u_n$.
\end{enumerate}
\end{definition}

Notice that ultrafilters are built over sequences of elements from~$A^+$, instead of elements of $A$, in order to match the semantics of $\ppml$. In particular, the construction does \emph{not} reduce to the standard $\bml$ construction for unimodal signatures, and it gives an infinite structure even if the original structure is finite. This is a particularity of the $\ppml$ setting. Interestingly, to the best of our knowledge, there are no other ultrafilter constructions in the literature based on sequences or strings.

Given a binary relation $R$ over $A$ and an ultrafilter $u$ over $A^+$, let
$$l_{R}^{-1}[u] \coloneqq \{Y \subseteq A^+ \mid l_{R}(Y)\in u\}.$$
In words, $l_{R}^{-1}$ is the preimage function $f^{-1}: \powerset(S') \to \powerset(S)$ for a function $f: S \to S'$ when $f = l_{R}$.
The following lemma will be useful to produce sucessors in an ultrafilter extension.

\begin{lemma}\label{lemma:carac-ufrelation}
Let $\modelA$ be a $\sigma$-structure and let $u$ and $v$ be ultrafilters over $A^+$. Then
\[
    (u,v) \in E^{\ueA}
    \quad\iff\quad
    l_{E^\modelA}^{-1}[u] \subseteq v.
\]
\end{lemma}

\begin{proof}
To simplify notation, in the context of this proof let us write $m$ and $l$ for $m_{E^{\modelA}}$ and $l_{E^{\modelA}}$, respectively, and for any $X \subseteq A^+$ let $X^c \coloneqq A^+\setminus X$. Begin by noticing that $(u,v)\in E^{\ueA}$ if and only if for all $X \in A^+$, $m(X)\not \in u$ implies $X \not \in v$ (contrapositive of the definition).

For $\Rightarrow$), suppose $m(X)\not \in u$ implies $X \not \in v$ for all $X \in A^+$. Then,
            $Y \in l_{E^\modelA}^{-1}[u]$ iff $l(Y)\in u$ iff (Prop.~\ref{prop:l&mProperties}(2))  $m(Y^c)^c \in u$ iff  $m(Y^c) \not\in u$ implies (hyp.) $Y^c \not\in v$ iff  $Y \in v$.

For $\Leftarrow$), take as hypothesis that $l_{E^\modelA}^{-1}[u]\subseteq v$. Then,
            $m(X) \not\in u$ iff $m(X)^c \in u$ 
            iff (Prop.~\ref{prop:l&mProperties}(2)) $l(X^c) \in u$ 
            iff $X^c \in l_{E^\modelA}^{-1}[u]$
            implies $X^c \in v$
            iff  $X \not\in v$.        \qedhere




\end{proof}

Let us now characterize the semantics of $\ppml$ formulas in an ultrafilter extension $\ueA$ in terms of their semantics in $\modelA$.




\begin{proposition}\label{prop:corresp-ultrafilterZ}
Let $\modelA$ be a $\sigma$-structure, let $\phi \in \ppml$, let $n \in \N$ with $n > \D{\phi}$, and let $\bar u = u_1,\dots,u_n$ be a sequence of ultrafilters over $A^+$. 
Then $\ueA,u_1\cdots u_n \models \phi$ iff  $\val{\phi}{\modelA}\in u_n$.
\end{proposition}

\begin{proof}
We simplify notation as in the proof of Lemma~\ref{lemma:carac-ufrelation} by writing $m$ and $l$ for $m_{E^{\modelA}}$ and $l_{E^{\modelA}}$, respectively, and $X^c$ for the complement in $A^+$ of any $X \subseteq A^+$. We also write $\ext{\phi}$ for the extension $\val{\phi}{\modelA}$ of any $\phi \in \ppml$. We proceed by induction on the structure of $\phi$. 

\noindent\textbf{Case $\phi=\top$.} It is always the case that $\bar u \models \top$ and $\ext{\top} = A^+ \in u$.

\noindent\textbf{Case $\phi=R$.} Let $r=\arity(R)$. Then
    $$ \bar u \models R \IFF \last_r(\bar u)\in R^{\ueA} \IFF r\leq n \AND \ext{R}\in u_n \IFF \ext{R} \in u_n.$$

\noindent\textbf{Case $\phi=\psi_1\land\psi_2$.} We have
    $$\bar u \models \psi_1\land\psi_2 \ \iff \ \bar u \models \psi_1 \AND \bar u \models \psi_2 \ \iff \ \ext{\psi_1} \in u_n \AND \ext{\psi_2}\in u_n \ \iff\ \ext{\psi_1\land\psi_2} = \ext{\psi_1}\cap\ext{\psi_2} \in u_n$$
    where the last ``if and only if'' follows from the fact that $u_n$, being a filter, is closed upwards and under finite intersections.

\noindent\textbf{Case $\phi=\neg\psi$.} We have $\bar u \models \lnot \psi$ iff $\bar u \not\models \psi$, iff $\ext{\psi}\not\in u_n$, iff $\ext{\lnot\psi} = \ext{\psi}^c \in u_n$.

\noindent\textbf{Case $\phi=\Diamond\psi$.} We have  $\bar u \models \Diamond \psi$ iff $\exists u \text{ ultrafilter on $A^+$ such that } u_n \prec u \AND \bar u u\models \psi$, iff $\exists u$ such that $\forall X \in u. m(X)\in u_n \AND \ext{\psi}\in u$, where in the last equivalence we used both the inductive hypothesis and the definition of $u_n \prec u$ in an ultrafilter extension. Now by Prop.~\ref{prop:l&mProperties}(1), this last statement readily implies $m(\ext{\psi}) = \ext{\Diamond \psi} \in u_n$; let us see that the converse implication holds as well.
    
We must produce a certain successor $u_n \prec u$ from the knowledge that $\Diamond\psi \in u_n$.
Consider the collection of subsets $X \coloneqq l^{-1}[u_n] \cup \{\ext{\psi}\}$. Recall that, by the well-known ultrafilter theorem, if $X$ has the finite intersection property (every finite collection of subsets in $X$ has finite intersection) then there exists an ultrafilter $u$ over $A^+$ with $X \subseteq u$. Let us see that $X$ indeed has the finite intersection property.

Let $\Delta = \{Y_1,\dots,Y_r\} \subseteq X$. If $\ext{\psi}\not\in \Delta$, then $\Delta \subseteq l^{-1}[u_n]$, i.e. $l(Y_i)\in u_n$ for all $i$, then $\cap_i l(Y_i) = l(\cap_i Y_i) \in u_n$, where we used Prop.~\ref{prop:l&mProperties}(3). This means that $l(\cap_i Y_i)\neq \emptyset$, which in turn implies $\cap_i Y_i \neq\emptyset$ (since $l(\emptyset) = \{s \in A^+ \mid E(s) \subseteq \emptyset\} = \emptyset$). Otherwise, if $\ext{\psi} \in \Delta$, then it is enough to show that $Y \cap \ext{\Diamond \psi} \neq\emptyset$ for all $Y \in l^{-1}[u_n]$. Given such $Y$, $l(Y) \in u_n$, then since by hypothesis $\ext{\Diamond\psi}\in u_n$, we get $l(Y)\cap \ext{\Diamond\psi} \in u_n$, which implies that the intersection is nonempty. Let $s \in l(Y)\cap \ext{\Diamond\psi}$. Then in particular $\exists s' \in E^\modelA(s)$ such that $s' \models \psi$; however $E^\modelA(s)\subseteq Y$ since $s \in l(Y)$, hence $s' \in Y \land \ext{\psi} \neq\emptyset$.

Since $X$ has the finite intersection property, there is an ultrafilter $u$ such that $X \subseteq u$. Since $u$ contains $l^{-1}[u_n]$, by Lemma~\ref{lemma:carac-ufrelation} we have $u_n \prec u$, and 
since $\ext{\psi}\in u$, by the inductive hypothesis applied on sequences of length $n+1$, we get that $\bar u u \models \psi$, which is to say $\bar u \models \Diamond \psi$.
\end{proof}

Prop.~\ref{prop:corresp-ultrafilterZ} might be surprising, since the satisfaction relation for an ultrafilter extension depends exclusively on the last ultrafilter in the valuation, assuming that the valuation is sufficiently long; however this can already be expected from how the ultrafilter extension is defined.

We remark in passing that there is a sense in which the ultrafilter extension $\ueA$ of a $\sigma$-structure $\modelA$ can be viewed as a Kripke model; moreover our definition of $\ueA$ can be understood as an application of the standard $\bml$ ultrafilter extension construction, but applied ``somewhere else'', i.e. in a Kripke model constructed from $\modelA$ in a certain way. This gives a conceptual explanation for Prop.~\ref{prop:corresp-ultrafilterZ}. We leave the exposition and exploration of this interpretation to future work.

Recall that given $s\in A^+$, the \defstyle{principal ultrafilter generated by $s$}, which we denote by $\pi_s$, is the smallest ultrafilter containing the singleton set $\{s\}$, i.e. $\pi_s = \{ X \subseteq A^+ \mid s \in X\}$. By identifying principal ultrafilters we can obtain the following invariance result.


\begin{lemma}\label{lemma:equivUF}
Let $\modelA=(A,\cdot^{\modelA})$ be a $\sigma$-structure and let $s \in A^+$. Then
\[
(\modelA,s) \equiv (\ueA,\bar u \pi_s),
\]
for every sequence $\bar u$ of ultrafilters such that $\card{\bar u} = \card{s}-1$.
\end{lemma}

\begin{proof}
    Let $n\coloneqq |s|$ and $\bar u$ be such that $|\bar u| = |s|-1$. Then by Lemma~\ref{lem:debt-normalization}, we want to prove that $(\modelA, s) \equiv_n (\ueA, \bar u \pi_s)$. Let $\phi \in \ppml$ be a formula with $\D{\phi}\leq n-1$, then
    $$
    \modelA, s \models \phi \quad \iff \quad
     s \in \val{\phi}{\modelA} \quad \iff \quad
      \val{\phi}{\modelA} \in \pi_s \quad \iff \quad
       \ueA,\bar u \pi_s \models \phi
    $$
    where the last equivalence follows from Prop.~\ref{prop:corresp-ultrafilterZ}.
\end{proof}

With this lemma at hand, we are able to show that ultrafilter extensions are a way to complete a model in order to obtain a saturated model. This is stated in the following property.

\begin{proposition} \label{prop:UFesSaturado}
Let $\modelA$ be a $\sigma$-structure. Then $\ueA$ is a $\ppml$-saturated model.
\end{proposition}
\begin{proof}
    Let $\Sigma$ be a set of formulas and let $\bar u = u_1\ldots u_n$ be a sequence of ultrafilters on $A^+$ such that $\Sigma$ is finitely satisfiable in the successors of $\bar u$, $E^{\ueA}(\bar u)$. We will show that there exists some ultrafilter $u^\prime$ such that $\bar u u^\prime$ that satisfies $\Sigma$. We first observe that, since we will only evaluate formulas from $\Sigma$ on sequences of length $n+1$, by Lemma~\ref{lem:debt-normalization} we may assume without loss of generality that all formulas $\phi \in \Sigma$ have $\D{\phi}\leq n$, replacing $\phi$ by its normalization $N_{n+1}(\phi)$.
    
    We now follow a similar argument as in the inductive step for $\Diamond$ in Prop.~\ref{prop:corresp-ultrafilterZ}. 
    Enumerate $\Sigma$ as $(\phi_i)_{i \in \mathbb{N}}$. For each $i \in \mathbb{N}$, let 
    $\psi_i := \phi_1 \land \cdots \land \phi_i$. We define
        $\Delta_1 \coloneqq  \{ \val{\psi_i}{\modelA} \mid i\in \mathbb{N} \}$,
        $\Delta_2 \coloneqq  \{ W \mid l_E(W) \in u \}$, and
        $\Delta \coloneqq \Delta_1 \cup \Delta_2$.
    Let us see that $\Delta$ has the finite intersection property. First observe that both $\Delta_1$ and $\Delta_2$ are closed under finite intersections and do not contain $\emptyset$. This implies that both $\Delta_1$ and $\Delta_2$ have the finite intersection property and, moreover, to show that $\Delta$ has it too it is enough to check that given $Y \in \Delta_1$ and $X \in \Delta_2$, $X\cap Y\neq \emptyset$.

    Given $X, Y$ as above, we have $Y = \val{\psi_n}{\modelA}$ for some $n\in\mathbb{N}$, hence there exists some ultrafilter $u^{\prime\prime}$ such that $u\prec u^{\prime\prime}$ and $\bar u u^{\prime\prime} \models \psi_n$. In particular, since $\D{\psi_n} \leq n$ we conclude by Prop.~\ref{prop:corresp-ultrafilterZ} that $Y \in u^{\prime\prime}$. On the other hand, from $u\prec u^{\prime\prime}$ it follows by Lemma~\ref{lemma:carac-ufrelation} that $X \in u^{\prime\prime}$. Hence $X\cap Y\in u^{\prime\prime}$ and thus the intersection is non-empty.
    
    Since $\Delta$ has the finite intersection property, applying the ultrafilter theorem we can extend $\Delta$ to an ultrafilter $u^{\prime}$. By analogous reasoning to that of the inductive step for $\Diamond$ in Prop.~\ref{prop:corresp-ultrafilterZ}, $u\prec u^\prime$ and $\bar u u^\prime$ satisfies $\Sigma$.
\end{proof}

Finally, we obtain the intended result: we shift model equivalence under valuation semantics, to `bisimilarity somewhere else', namely, to bisimilarity for sequences over ultrafilter extensions. 

\begin{theorem}\label{th:equivCharUF}
Let $\modelA=(A,\cdot^{\modelA})$ and $\modelB=(B,\cdot^{\modelB})$ be two models, let $s\in A^+$ and $t\in B^+$, and let $\bar u$ and $\bar v$ be sequences of ultrafilters over $A^+$ and $B^+$, respectively, such that
$
\card{\bar u}=\card{s}-1
$ and $
\card{\bar v}=\card{t}-1.
$
Then
\[
(\modelA,s) \equiv (\modelB,t)
\quad \iff \quad 
(\ueA,\bar u \pi_s) \bisimilar (\ueB,\bar v \pi_t).
\]
\end{theorem}

\begin{proof}
By Lemma~\ref{lemma:equivUF}, we have
$
(\modelA,s)\equiv(\ueA,\bar u\pi_s)
$ and $
(\modelB,t)\equiv(\ueB,\bar v\pi_t).
$
Hence
$
(\modelA,s)\equiv(\modelB,t)
$ iff $
(\ueA,\bar u\pi_s)\equiv(\ueB,\bar v\pi_t).
$
By Prop.~\ref{prop:UFesSaturado}, both $\ueA$ and $\ueB$ are $\ppml$-saturated. Therefore, by Thm.~\ref{th:saturationImplHM}, the class of $\ppml$-saturated structures has the Hennessy--Milner property, and so
$
(\ueA,\bar u\pi_s)\equiv(\ueB,\bar v\pi_t)
$ iff $
(\ueA,\bar u\pi_s)\bisimilar(\ueB,\bar v\pi_t).
$
This yields the result.
\end{proof}




%% file: characterization.tex
\section{Van Benthem-style Characterization Theorem}
\label{sec:characterization}

We now turn to answer one of the classical and fundamental questions in model theory for modal logics: what is the fragment of first-order logic captured by $\ppml$? To answer such a question, we will follow the path of the \emph{van Benthem characterization theorem}~\cite{vanBenthem1983}, for which bisimulations play a crucial role. In this regard, we will establish that the standard translation for $\ppml$ corresponds to the $\fol$ fragment that is invariant under bisimulation. This helps us to characterize exactly the expressive power of $\ppml$ in terms of $\fol$. 

In order to obtain our main theorem, we will need to quickly review and adapt several notions and preliminary results. First, we will introduce a class of models that constitutes a  Hennessy--Milner class: the class of $\omega$-saturated models. The advantage of $\omega$-saturated models is that they can be constructed via \emph{ultraproducts}, a classical tool in model theory. 
Then, we will establish a correspondence between logical equivalence of models and bisimilarity of their ultrapowers, a result that will be ultimately fundamental in the main theorem. 

In what follows, we will assume some basic knowledge of model theory, e.g., notions like ultraproducts or expansions will not be introduced here. In particular, for a $\sigma$-structure $\modelA=(A,\cdot^{\modelA})$, we denote by $\prod_U \modelA$ its ultrapower modulo some ultrafilter $U$, with $\prod_U A$ its domain. For further details, see, e.g.,~\cite{Chang&KeislerMT}.


\begin{definition}\label{def:realizes}
A $\sigma$-structure  $\modelA$ is \defstyle{$\omega$-saturated} if for every finite subset $Y \subseteq A$, the expansion $\modelA_Y$ realizes every set of formulas $\Gamma(x)$ in $\mathcal{L}_Y$ that is consistent with the theory of $\modelA_Y$.
\end{definition}

The next property positions $\omega$-saturated models as an interesting class for $\ppml$.

\begin{theorem}\label{th:omegaSatImplyMsat}
Every $\omega$-saturated model is $\ppml$-saturated. Therefore, the class of $\omega$-saturated models has the Hennessy--Milner property.
\end{theorem}

\begin{proof}
Let $\modelA$ be an $\omega$-saturated $\sigma$-structure (seen as a $\fol$-model). Let $s = s_1 \ldots s_n \in A^+$ and let $\Sigma$ be a set of $\ppml$-formulas. To show that $\modelA$ is $\ppml$-saturated, assume that $\Sigma$ is finitely satisfiable in $E^\modelA(s)$. 

We define the following set of $\fol$-formulas with at most one free variable in the expanded model $\modelA_{\set{s_1,\dots,s_n}}$ as 
    $\Sigma' = \{ E(\bar{s_n},x) \} \cup \ST_{x}(\Sigma)$, 
where
$
\ST_{x}(\Sigma) = \{\, \ST_{x_1\dots x_n x}(\phi)(\bar{s_1},\ldots,\bar{s_n},x) \mid \phi \in \Sigma \,\},
$
and each $\bar{s_i}$ is a constant symbol added to the language and interpreted as $s_i$. 
We show that this set is consistent with the theory of $\modelA_{\set{s_1,\dots,s_n}}$. Let $\Delta \subseteq \Sigma'$ be finite. By hypothesis, the subset of $\Delta$ consisting of formulas from $\ST_{x}(\Sigma)$ is satisfiable at some $a$ s.t. $s_n \prec a$, and if $E(\bar{ s_n},x)$ belongs to $\Delta$, then the same element also satisfies $E(\bar{s_n},x)$. 
Hence, by compactness, $\Sigma'$ is satisfiable and therefore consistent with the theory of $\modelA_{\set{s_1,\dots,s_n}}$.

Finally, since $\modelA$ is $\omega$-saturated, there exists an element $a$ realizing $\Sigma'$. By construction of $\Sigma'$, $a$ is such that $s_n \prec a$ and it satisfies $\ST_{x}(\Sigma)$. By~Prop.~\ref{prop:equiv1orderPPML}, we get that $\modelA, sa \models \Sigma$, thus $\modelA$ is $\ppml$-saturated.
\end{proof}

Next, we state the so-called {\L}o\'s's Theorem, establishing that satisfaction of $\fol$-formulas is preserved under ultraproducts. As usual, this will be a fundamental tool in our characterization result.

\begin{theorem}[{\L}o\'s's Theorem] \label{th:Los}
Let $U$ be an ultrafilter over a nonempty set $I$ and for each $i \in I$, let $\modelA_i$ be a $\sigma$-structure, with $\prod_U \modelA_i$ their ultraproduct modulo $U$. 
Then, for any $\fol$-formula $\psi(x_1,\ldots,x_n) \in \mathcal{L}$ and any $f_U^1,\ldots,f_U^n \in \prod_U A_i$, we have:
\[
    \prod_U \modelA_i \models \psi(x_1,\ldots,x_n)[f_U^1,\ldots,f_U^n] 
     \ \iff \ 
    \{ i \in I \mid \modelA_i \models \psi(x_1,\ldots,x_n)[f^1(i),\ldots,f^n(i)] \} \in U.
\]
\end{theorem}

Thus, we obtain the corresponding result for $\ppml$.

\begin{corollary} \label{cor:cteUltrapower}
Let $\modelA=(A,\cdot^{\modelA})$ be a $\sigma$-structure, and let 
$U$ be an ultrafilter. For $a\in A$, let $(f_a)_U \in \prod_U A$ be the constant function $f_a(i) = a$ for every $i \in I$, and for $t=a_1\ldots a_n\in A^+$, let $f_t$ be the function  consisting of the sequence of functions $f_{a_1}\ldots f_{a_n}$. Then, for every $\ppml$-formula $\psi$ and every $s \in A^+$ such that $\D{\psi}\leq \card{s}-1$, we have
\[
        \modelA, s \models \psi \ \iff \ \prod_U \modelA, (f_s)_U \models \psi.
\]
\end{corollary}

\begin{proof}
Let $\psi$ be a $\ppml$-formula and let $\bar x = x_1 \ldots x_n$ be a sequence of $\fol$ variables. Applying the standard translation to the right-hand side of the equivalence that we want to prove, and then applying~Thm.~\ref{th:Los}, we obtain
\[ 
\prod_U \modelA \models \ST_{\bar x}(\psi)[(f_s)_U]
\ \iff \ 
\{ i \in I \mid \modelA \models \ST_{\bar x}(\psi)[f_s(i)] \} \in U.
\]
Since $f_s$ is constant, this is equivalent to
\[  
\prod_U \modelA \models \ST_{\bar x}(\psi)[(f_s)_U]
\ \iff \ 
\{ i \in I \mid \modelA \models \ST_{\bar x}(\psi)[s] \} \in U.
\]
By applying Prop.~\ref{prop:equiv1orderPPML} to the right-hand side we have the following statement, which is equivalent to what we want to prove:
\[
\prod_U \modelA, (f_s)_U \models \psi
\ \iff \ 
\{ i \in I \mid \modelA, s \models \psi \} \in U.
\]
Observe that, if $\modelA, s \models \psi$, then $\{ i \in I \mid \modelA, s \models \psi \} = I$, 
which trivially belongs to $U$. Otherwise,
$\{ i \in I \mid \modelA, s \models \psi \} = \emptyset$, 
which trivially does not belong to $U$. Hence, 
$\modelA, s \models \psi
\ \iff \ 
\prod_U \modelA, (f_s)_U \models \psi$.
\end{proof}

We need to recall a few more technical concepts to complete our result. 
We say that a filter $F$ is \textbf{countably incomplete} if and only if there exists a countable set $E \subseteq F$ such that $\bigcap E \not \in F$.

The following result is well known \cite{Chang&KeislerMT}:

\begin{proposition} \label{prop:incompleteUltrafilter}
Let $\mathcal{L}$ be a countable first-order language, let $U$ be a countably incomplete ultrafilter over a nonempty set $I$, and let $\foModelM$ be an $\mathcal{L}$-structure. Then the ultraproduct $\prod_U \foModelM$ is $\omega$-saturated.
\end{proposition}

It is time to combine all our ingredients. To do so, we introduce the so-called \emph{Detour theorem}, but specifically instantiated with $\ppml$-formulas. This result enables us to establish a relation between logical equivalence of models, and bisimilarity at the level of utraproducts.

\begin{theorem} \label{th:detour}
Let $\modelA$ and $\modelB$ be $\sigma$-structures, and let $s \in A^+$ and $t \in B^+$ such that $\card{s}=\card{t}=n$. Then, the following statements are equivalent:
\begin{enumerate}
    \item $(\modelA,s) \equiv (\modelB,t)$.
    \item There exist ultrapowers $\prod_U \modelA$ and $\prod_U \modelB$ modulo some ultrafilter $U$ such that $(\prod_U \modelA, (f_s)_U)\bisimilar (\prod_U \modelB,(f_t)_U)$.
\end{enumerate}
\end{theorem}
\begin{proof}
    Let us start by proving $2) \Rightarrow 1)$, then assume there exist ultrapowers $\prod_U \modelA$ and $\prod_U \modelB$ modulo some ultrafilter $U$ such that $(\prod_U \modelA, (f_s)_U)\bisimilar (\prod_U \modelB,(f_t)_U)$. Notice that, by~Thm.~\ref{th:bisImplyEquiv}, we have $(\prod_U \modelA, (f_s)_U)\equiv (\prod_U \modelB,(f_t)_U)$. In particular, $(\prod_U \modelA, (f_s)_U)\equiv_n(\prod_U \modelB,(f_t)_U)$. 
    Under this hypothesis, we can apply~Cor.~\ref{cor:cteUltrapower} and obtain $\modelA,s\equiv_n\modelB,t$. But by~Lemma~\ref{lem:debt-normalization}, we get  $\modelA,s\equiv\modelB,t$.

    To show $1) \Rightarrow 2)$, assume $(\modelA,s) \equiv (\modelB,t)$. By hypothesis, we get that in particular $\modelA,s\models\psi$ iff $\modelB,t\models\psi$, for all $\psi$ such that $\D{\psi}\leq n-1$. Let $U$ be a countably incomplete ultrafilter over $\mathbb{N}$.\footnote{It is well known that such ultrafilters exist. One example is the following: consider the filter $\{ X \subseteq \mathbb{N} \mid \mathbb{N} \setminus X \text{ is finite} \}$, then extend it to an ultrafilter using the ultrafilter theorem.} By Cor.~\ref{cor:cteUltrapower}, we get  $(\prod_U \modelA, (f_s)_U)\equiv_n(\prod_U \modelB,(f_t)_U$. By Cor.~\ref{cor:cteUltrapower}, and therefore by~Lemma~\ref{lem:debt-normalization}, we have  $(\prod_U \modelA, (f_s)_U)\equiv (\prod_U \modelB,(f_t)_U)$. 
    Moreover, by Prop.~\ref{prop:incompleteUltrafilter}, $\prod_U \modelA$ and $\prod_U \modelB$ are $\omega$-saturated. Hence, by Thm.~\ref{th:omegaSatImplyMsat}, these models satisfy the Hennessy--Milner property, thus we can conclude that $(\prod_U \modelA, (f_s)_U)\bisimilar (\prod_U \modelB,(f_t)_U)$.
\end{proof}

In what follows, let $\alanguage_\sigma^n$ denote the first-order language over signature $\sigma$ with $n$ free variables. 

\begin{definition}\label{def:invariant-under-bisim}

We say that a formula $\phi(\bar{x}) \in \mathcal{L}_\sigma^n$ is invariant under bisimulation if for every bisimulation $Z$ between $\modelA$ and $\modelB$, and for all $s \in A^+$ and $t \in B^+$ such that $\card{s}=\card{t}=n$ and  $sZt$, we have
\[
\modelA \models \phi(s) \quad \iff\quad  \modelB \models \phi(t).
\]
\end{definition}

Finally, we are able to prove the result usually known as \emph{van Benthem characterization theorem}, which makes explicit the connection between $\ppml$, $\fol$ and $\ppml$-bisimulations.

\begin{theorem}\label{th:vanBenthem}
Let 
$\phi(\bar{x}) \in \mathcal{L}_\sigma^n$. Then $\phi$ is invariant under bisimulation if and only if  
$\models\phi(\bar{x}) \lra \ST_{\bar{x}}(\psi)$, for some $\psi \in \ppml$ such that $\D{\psi} \leq n-1$. 
\end{theorem}
\begin{proof}
    The nontrivial case is the left-to-right direction, since the converse follows directly from Thm.~\ref{th:bisImplyEquiv}. 
    Suppose that $\phi$ is invariant under bisimulation. Consider the set of $\ppml$-consequences of $\phi$:
    \begin{align*}
        \moc(\phi) := \{  \ST_{\bar{x}}(\psi) \mid \psi \in \ppml \text{ such that  } \phi(\bar{x}) \models \ST_{\bar{x}}(\psi) \}.
    \end{align*} 
    
    First, let us show that, if $\moc(\phi) \models \phi(\bar{x})$, then $\phi(\bar{x})$ is equivalent to the translation via $\ST$ of a $\ppml$-formula with debt bounded by $n-1$. Notice that
    by the compactness theorem there is a finite $\Delta \subseteq \moc(\phi)$ such that $\Delta \models \phi(\bar{x})$. Equivalently, 
        $\models \bigwedge_{\ST_{\bar{x}}(\chi) \in \Delta} \ST_{\bar{x}}(\chi) \to \phi(\bar{x})$, 
    and, since  $\Delta$ is finite, $ \bigwedge_{\ST_{\bar{x}}(\chi) \in \Delta} \ST_{\bar{x}}(\chi)$ is a $\fol$-formula. Moreover, by the definition of $\moc(\phi)$, the converse also holds, hence 
        $\models \bigwedge_{\ST_{\bar{x}}(\chi) \in \Delta} \ST_{\bar{x}}(\chi) \leftrightarrow \phi(\bar{x})$.
    Thus, since $ \bigwedge_{\ST_{\bar{x}}(\chi) \in \Delta} \ST_{\bar{x}}(\chi)$ is the standard translation of a $\ppml$-formula. Finally, by Lemma~\ref{lem:debt-normalization}, we obtain the desired $\ppml$ formula $\psi$ such that $\D{\psi} \leq n - 1$.

    Therefore it is enough to prove that $\moc(\phi)\models\phi(\bar{x})$. Let $\foModelM$ be a first-order model such that $\foModelM \models \moc(\phi)[s]$, for some $s\in M^n$. We must show that $\foModelM \models \phi(\bar{x})[s]$. 

    To use the hypothesis that $\phi$ is bisimulation-invariant we will construct a model $\foModelN$ and a tuple $ t\in N^n$ such that $\foModelN\models\phi[t]$ and $(\foModelM,s)\bisimilar(\foModelN,t)$. The idea is to build a theory that, besides containing $\phi$, ensures that any two elements satisfying it are $\ppml$-equivalent. If we succeed, we can use our previous results to ``jump'' to saturated models, where equivalence and bisimulation coincide, and then apply the bisimulation-invariance of $\phi$. 

    With this in mind, define:
$T(\bar{x}) = \{ \ST_{\bar{x}}(\psi) \mid \psi \in \ppml \text{ and } \foModelM \models \ST_{\bar{x}}(\psi)[s] \}.$ 
    This set has the required property: any tuple satisfying $T(\bar{x})$ will be $\ppml$-equivalent to $s$. Hence we look for a model $\foModelN$ and a tuple $t$ such that $\foModelN \models T(\bar{x}) \cup \{\phi(\bar{x})\}[t]$. We prove the existence of such a model by showing that $T(\bar{x})\cup\{\phi(\bar{x})\}$ is consistent. Indeed, suppose towards a contradiction that $T(\bar{x})\cup\{\phi(\bar{x})\}$ is inconsistent. Then $T(\bar{x})\not\models \phi(\bar{x})$, i.e. $T(\bar{x})\models \neg\phi(\bar{x})$. By compactness there is a finite $T_0(\bar{x})\subseteq T(\bar{x})$ such that 
        $\models \Big(\bigwedge_{\chi\in T_0(\bar{x})} \chi\Big) \rightarrow \neg\phi(\bar{x})$. 
    Equivalently, by contrapositive we get that 
        $\models \phi(\bar{x}) \rightarrow \neg\Big(\bigwedge_{\chi\in T_0(\bar{x})} \chi\Big)$, 
    thus $\neg\big(\bigwedge_{\chi\in T_0(\bar{x})} \chi\big)\in \moc(\phi)$. But then, since $\foModelM \models \moc(\phi)[s]$, we would have $\foModelM \models \neg\big(\bigwedge_{\chi\in T_0(\bar{x})} \chi\big)$, contradicting $\foModelM \models T(\bar{x})$. This contradiction shows that $T(\bar{x})\cup\{\phi(\bar{x})\}$ is consistent.

    Hence, there exists a model $\foModelN$ and a tuple $t\in N^n$ such that $\foModelN \models T(\bar{x})\cup\{\phi(\bar{x})\}[t]$. By the definition of $T(\bar{x})$, it follows that $\foModelM,s \equiv \foModelN,t$, by seeing $\foModelM$ and $\foModelN$ as models of $\ppml$. Then, by~Thm.~\ref{th:detour} there is an ultrafilter $U$ such that $(\prod_U \foModelM,(f_s)_U) \bisimilar (\prod_U \foModelN,(f_t)_U)$. Since $\foModelN \models \phi(\bar{x})[t]$, by~Thm.~\ref{th:Los}  
    we get $\prod_U \foModelN \models \phi(\bar x)[(f_t)_U]$. Because $\phi(\bar{x})$ is invariant under bisimulation, it follows that $\prod_U \foModelM \models  \phi(\bar x)[(f_s)_U]$, and hence $\foModelM \models \phi(\bar{x})[s]$, completing the proof.
\end{proof}

%% file: final.tex
\section{Conclusions}
\label{sec:final}

This paper continues the model-theoretic study of Path Predicate Modal Logic ($\ppml$), complementing the categorical tools by which this study was initiated in~\cite{FigGoren25}. Here, we deal with more traditional model-theoretic machinery such as saturation, ultrafilter extensions and ultraproducts. We started by introducing an unbounded notion of bisimulation for valuations, i.e. sequences of states, showing that bisimilar valuations satisfy the same $\ppml$-formulas. Then, we identified natural Hennessy--Milner classes for the logic, and proved that ultrafilter extensions provide canonical saturated companions. Finally, we established a van Benthem-style characterization theorem, showing that $\ppml$ corresponds exactly to the fragment of $\fol$ that is invariant under bisimulation.

A distinctive feature of $\ppml$ is that, unlike $\bml$, it is not naturally evaluated at a single point, but rather at a finite sequence of points. This difference is reflected throughout the paper. In $\bml$, the standard translation maps formulas to first-order formulas with one free variable, corresponding to the current state of evaluation. By contrast, in $\ppml$ the natural first-order counterpart uses tuples of free variables representing an $E$-chain, since relation symbols are interpreted on states occurring along a path. Accordingly, the notions of bisimulation and model-theoretic constructions considered in this paper must be formulated over sequences rather than single states.

This valuation-based perspective suggests several directions for further research. 
To start with, it would be interesting to identify a first-order fragment that captures more faithfully the specific shape of the standard translation of $\ppml$. As discussed in Sec.~\ref{sec:st}, the image of the translation sits naturally inside the fluted fragment, but this fragment is still substantially more expressive than what is needed for $\ppml$. The formulas arising from $\ppml$ obey not only the fluted ordering of variables, but also a strict path discipline governed by the distinguished relation $E$: quantification extends the current valuation one step at a time along an $E$-chain, and atomic predicates are evaluated only on contiguous suffixes of that valuation. This suggests the study of a more refined fragment of $\fol$, combining flutedness with an explicit $E$-chained, or path-guarded, discipline. A precise analysis of such a fragment could help clarify both the exact first-order nature of $\ppml$ and the source of its comparatively tame complexity.

Closely related to this is the possibility of relaxing the fluted discipline already at the modal level. In the present formulation of $\ppml$, relation symbols are always evaluated on contiguous suffixes of the current valuation, so the variables are used in a strictly ordered, fluted fashion. A natural extension would be to allow formulas to refer to previously visited points in a less rigid way, without preserving the suffix order. Studying such a non-fluted variant of $\ppml$ could help clarify which aspects of the logic depend essentially on the fluted organization of variables, and which belong more fundamentally to its path-based semantics.

Another natural continuation of the present work is to pursue further model-theoretic and expressivity-theoretic results for $\ppml$, including questions of interpolation and Beth definability. Another line of research is to investigate proof-theoretic aspects of the logic, such as complete axiomatizations for well-nested fragments. It would also be interesting to study extensions of $\ppml$ with richer modal resources---for example,  fixpoint operators, or additional path constructors---and determine to what extent the model-theoretic picture developed here persists. Finally, it would be worthwhile to clarify further the connections between $\ppml$ and data-aware query languages (including GQL~\cite{FrancisGGLMMMPR23b} and SHACL~\cite{Ortiz23}), and to identify other natural logical formalisms that can be uniformly captured within this path-based framework.